\documentclass[letter,twocolumn]{article}

\usepackage{graphicx}
\usepackage{times}
\usepackage{epsfig}
\usepackage{amsmath}
\usepackage{amssymb}
\usepackage{algorithm}
\usepackage{algorithmic}
\usepackage{color}
\usepackage{cite}
\usepackage{fix2col}
\newcommand{\eq}{\triangleq}

\newcommand{\hfs}{\hfill\ensuremath{\square}} 
\DeclareMathOperator*{\argmin}{arg\,min}
\DeclareMathOperator*{\supp}{supp}

\newcommand{\field}[1]{\mathbb{#1}}
\newcommand{\R}{\field{R}}
\newcommand{\N}{\field{N}}

\newcommand{\F}{{\mathcal{F}}}
\newcommand{\E}{{\mathcal{E}}}
\newcommand{\T}{\top}
\newcommand{\rank}{\mathrm{rank}}
\newcommand{\diag}{\mathrm{diag}}
\newcommand{\vc}[1]{{\boldsymbol{#1}}}
\newcommand{\K}{K}
\newcommand{\Bs}{B}

\newcommand{\vx}{\vc{x}}
\newcommand{\vu}{\vc{u}}
\newcommand{\vz}{\vc{0}}
\newcommand{\vv}{\vc{v}}
\newcommand{\eps}{\vc{\varepsilon}}
\newcommand{\Sr}{{\mathcal{S}}}
\newcommand{\vr}{\vc{r}}
\newcommand{\vg}{\vc{g}}
\newcommand{\OMP}{\text{OMP}}

\newcommand{\elll}{\ell^1/\ell^2}
\newcommand{\Pzero}{{\bf P}_0}

\newcommand{\Qone}{{\bf Q}_1}
\newcommand{\Qtwo}{{\bf Q}_2}

\newtheorem{thm}{Theorem}[section]
\newtheorem{ass}[thm]{Assumption}

\newtheorem{lem}[thm]{Lemma}

\newtheorem{defn}[thm]{Definition}

\newtheorem{rem}[thm]{Remark}

\title{Packetized Predictive Control for Rate-Limited Networks via
Sparse Representation%
\thanks{This research was supported in part under 
the MEXT (Japan) Grant-in-Aid for Young Scientists (B) No.~22760317,
and also
 Australian Research Council's
  Discovery Projects funding scheme (project number DP0988601).} 
}

 \author{Masaaki Nagahara%
  \thanks{Masaaki Nagahara is with the Graduate School of Informatics, Kyoto
    University, Kyoto, 606-8501, 
    Japan; e-mail: \texttt{nagahara@ieee.org}. 
   },
   Daniel~E.~Quevedo%
   \thanks{
       Daniel Quevedo is  with the School of
    Electrical Engineering \& 
      Computer Science, The University of Newcastle, NSW
      2308, Australia; e-mail:
      \texttt{dquevedo@ieee.org}. 
   },
   Jan~\O stergaard%
  \thanks{Jan \O stergaard is with the Department of Electronic
      Systems, Aalborg 
  University, Denmark; e-mail: \texttt{janoe@ieee.org}
  }
  }
 \date{}
\begin{document}

\maketitle

\begin{abstract}
We study a networked control architecture for linear time-invariant 
plants 
in which an unreliable data-rate limited network is placed
between the controller and the plant input.
The distinguishing aspect of the 
situation at hand is that an unreliable data-rate limited network is placed
between controller and the plant
input.
To achieve robustness with
respect to  dropouts, the controller transmits  data
packets  
containing plant input predictions, which minimize a finite horizon
cost 
function. In our formulation, we design sparse packets for rate-limited
networks, by adopting an an $\ell^0$ optimization, which can be effectively
solved by an 
orthogonal matching pursuit method. Our formulation ensures asymptotic stability of
the control loop  in the presence
of bounded packet dropouts. 
Simulation results indicate that the proposed
controller provides sparse control packets, 
thereby giving bit-rate reductions 
for the case of memoryless scalar coding schemes when 
compared to the use of, more common, quadratic cost  functions,
as in linear quadratic (LQ) control.
\end{abstract}

\section{Introduction}
\label{sec:introduction}
In networked control systems (NCSs) communication between controller(s) and
plant(s) is made through
unreliable and rate-limited communication links
such as wireless networks and the Internet; see e.g.,
\cite{ZhaBraPhi01,HesNagYon07,BemHeeJoh}
Many interesting challenges arise and successful NCS design
methods need to 
consider both control and communication aspects. In particular,
so-called {\em packetized predictive control} (PPC) has been shown to have
favorable stability and performance properties, especially in the presence of packet-dropouts
\cite{Bem98,CasMosPap06,tansil07,munchr08,QueNes11,queost11,pinpar11}.  
In PPC,
the controller output is obtained through solving 
a finite-horizon cost function on-line in a  receding horizon manner.
Each control {\em packet}  contains a sequence
of tentative plant inputs for a finite horizon of
future time instants and is transmitted through a communication channel. 
Packets which are successfully received at
the plant actuator side, are stored in a buffer to be 
used whenever later packets are dropped. 
When there are no
packet-dropouts, PPC reduces to model predictive control. For PPC to give desirable
closed loop properties,  the more unreliable the network is, the larger 
the horizon length (and thus the number of tentative plant input values
contained in each packet) needs to be chosen. Clearly, in principle, this would require increasing
the network bandwidth (i.e., its bit-rate), unless the transmitted signals are suitably compressed.

\par To address the compression issue mentioned above, in the present work we investigate the use of 
{\em sparsity-promoting optimizations} for PPC. Such techniques have been
widely studied in the recent signal processing literature in the context of
{\em compressed sensing} (aka {\em compressive sampling})
\cite{Don06,Can06,CanWak08,Mal,Ela,StaMurFad}.
The aim of compressed sensing is to reconstruct {a signal
from a small set of linear combinations of the signal}
by assuming that the original signal is sparse.
The core idea used in this area is to introduce a sparsity index
in the optimization. To be more specific,
the sparsity index of a vector $\vc{v}$ is defined by the amount of
nonzero elements in $\vc{v}$ and is usually denoted by $\|\vc{v}\|_0$,
called the ``$\ell^0$ norm.''
The compressed sensing problem is then formulated by an
$\ell^0$-norm optimization, which, being  combinatorial is, in  principle hard to solve \cite{Nat95}. Since sparse vectors
contain many 0-valued elements, they can be easily 
compressed by only coding a few nonzero values and their locations.
A well-known example of this kind of sparsity-inducing compression is JPEG
\cite{Wal91}. 

\par The purpose of this work is to adapt sparsity concepts for use in NCSs
over erasure channels. 
A key difference between standard  compressed sensing applications and NCSs is that
the latter operate in closed loop. Thus, time-delays need to be avoided and
stability issues studied, see also\cite{nagque11a}. 
To keep time-delays bounded, we adopt an 
iterative greedy algorithm called
{\em Orthogonal Matching Pursuit} (OMP)
\cite{MalZha93,PatRezKri93} for the on-line design of control packets. The algorithm is very simple and
known to be dramatically faster than exhaustive search. In relation to
stability in the
presence of bounded packet-dropouts,  our results show how to design the  cost
function to ensure  asymptotic stability of the  NCS. 

\par Our present manuscript complements our recent conference
contribution\cite{nagque11a}, which adopted an $\ell^1$-regularized $\ell^2$
optimization for PPC. A limitation of the approach in\cite{nagque11a} is that for
open-loop unstable systems,  asymptotic stability cannot be obtained  in the
presence of bounded  packet-dropouts; the best one can
hope for is   practical stability. Our current paper also complements the
extended abstract \cite{nagque12b}, by considering bit-rate issues and also
presenting a  detailed
technical analysis of the scheme, including proofs of results.  To the best of our knowledge,  the only
other published works which deal with sparsity and compressed sensing for
control are \cite{BhaBas11} which studies compressive
sensing for  state reconstruction in feedback  systems, and
\cite{nagque11a,nagque12a} which focus on sampling and command generation for remote applications.

\par The remainder of this work is organized as follows:
Section~\ref{sec:plant-model-control} revises basic elements of packetized
predictive control. 
In Section \ref{sec:design}, we formulate the design of the sparse control packets
in PPC
based on sparsity-promoting optimization.
In Section \ref{sec:stability}, we study stability of the resultant networked
control system.
Based on this,
in Section \ref{sec:relax}
we propose relaxation methods to compute sparse control packets
which leads to asymptotic (or practical) stability.
A numerical example is included in 
Section \ref{sec:examples}. Section \ref{sec:conclusion} draws conclusions.

\subsubsection*{Notation:}
We write $\N_0$ for $\{0, 1, 2, 3, \ldots\}$, $|\cdot |$ refers to
modulus 
of a
number. The identity
matrix (of appropriate dimensions) is denoted via $I$. 
For a matrix (or a vector) $A$, $A^\T$ denotes the transpose.
For a vector $\vv=[v_1,\ldots,v_n]^\T\in\R^n$ and a positive definite matrix $P>0$, 
we define
\[
 \|\vv\|_P := \sqrt{\vv^\T P \vv},~\|\vv\|_1:=\sum_{i=1}^n |v_i|,~\|\vv\|_\infty:= \max_{i=1,\ldots,n} |v_i|
\]
and also denote $\|\vv\|_2:= \sqrt{\vv^\T \vv}$.
For any  matrix $P$, $\lambda_{\max}(P)$ and $\lambda_{\min}(P)$
 denote the maximum and the minimum eigenvalues of $P$, respectively; $\sigma_{\max}^2(P):=\lambda_{\max}(P^\T P)$.

\section{Packetized Predictive Networked Control}
\label{sec:plant-model-control}
We consider 
discrete-time (LTI) plants with a scalar input:
\begin{equation}
  \label{eq:plant}
  \begin{split}
   \vx(k+1)&=A\vx(k)+\Bs u(k)+\vv(k), \quad k\in\N_0,\\
   \vx(0)&=\vx_0,
  \end{split}
\end{equation}
where $\vx(k)\in\R^n$, $u(k)\in\R$ and $\vv(k)\in\R^n$ is the plant
noise. Throughout this work, we assume that the pair $(A,\Bs)$ is reachable.

We are interested in an NCS architecture where the controller
communicates with the plant actuator through an erasure channel, see Fig.~\ref{fig:NCS}.
\begin{figure}[tbp]
\centering
\includegraphics{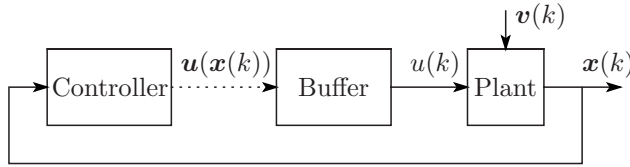}
\caption{NCS with PPC. The dotted line indicates an erasure
  channel.}
\label{fig:NCS}
\end{figure}
This channel introduces 
packet-dropouts, which we model via the
dropout sequence $\{d(k)\}_{k \in \N_0}$ in:
\[
d(k)\eq\left\{
\begin{array} {ll}
1, & \textrm{if packet-dropout occurs at instant $k$,}\\
0, & \textrm{if packet-dropout does not occur at time $k$.}\\
\end{array}\right .
\]

With PPC, as
described, for instance, in \cite{QueNes11},
 at each time instant $k$, the controller uses the state $\vx(k)$ of the plant
(\ref{eq:plant}) to calculate and
send a control packet of the form
\begin{equation}
  \label{eq:packet}
  \vu\bigl(\vx(k)\bigr)\eq
    \bigl[
	u_0\bigl(\vx(k)\bigr), u_1\bigl(\vx(k)\bigr),\dots, u_{N-1}\bigl(\vx(k)\bigr)
    \bigr]^\T\!\!\in\R^N
\end{equation}
to the plant input node. 

To achieve robustness
against packet dropouts, buffering is used.
More precisely, suppose that at time instant $k$, we have $d(k)=0$, i.e.,
the  
data packet $\vu\bigl(\vx(k)\bigr)$  is successfully received at
the plant input side. Then, this packet is stored in a buffer, overwriting its
previous contents.
If the next packet $\vu\bigl(\vx(k+1)\bigr)$ is dropped,
then the plant input $u(k+1)$ is set to $u_1\bigl(\vx(k)\bigr)$, the second element of
$\vu\bigl(\vx(k)\bigr)$.
The elements of $\vu\bigl(\vx(k)\bigr)$ are then successively used until
some packet $\vu\bigl(\vx(k+\ell)\bigr)$, $\ell \geq 2$
is successfully received. 

\section{Design of Sparse Control Packets}
\label{sec:design}
In PPC discussed above, the control packet $\vu\bigl(\vx(k)\bigr)$ is transmitted
at each time $k\in\N_0$ through an erasure channel (see Fig.~\ref{fig:NCS}).
It is often the case that the bandwidth of the channel is limited,
and hence one has to compress control packets to a smaller data size, see also\cite{QueOstNes10}.
To design packets which are easily compressible, we adapt techniques used in the
context of {\em compressed sensing} \cite{Don06,Can06} to design sparse 
control vectors $ \vu\bigl(\vx(k)\bigr)$. Since  sparse vectors contain many
$0$-valued elements, they
can be highly compressed by only coding their  few nonzero components and
locations, as will be illustrated in Section~\ref{sec:examples}.
Thus, the control objective in this paper is to find 
{\em sparse} control packets $\vu\bigl(\vx(k)\bigr)$
which ensure that 
the NCS with bounded packet dropouts is
{\em asymptotically stable}.

\par We define the sparsity of a vector $\vu$  by its $\ell^0$ ``norm,''
\[
 \|\vu\|_0 \eq \text{~the amount of nonzero elements in~} \vu\in\R^N
\]
and introduce the following sparsity-promoting optimization:
\begin{equation}
 \begin{split}
 \vu(\vx)&\eq \argmin_{\vu\in\R^N} \|\vu\|_0\\
 &\text{subject to}~
  \|\vx'_N\|_P^2 + \sum_{i=1}^{N-1}\|\vx'_i\|_Q^2 \leq \vx^\top W \vx,
 \end{split}
 \label{eq:opt0}
\end{equation}
where we omit the dependence on $k$, and
\[
 \begin{split}
   \vx'_0 &= \vx,\quad \vx'_{i+1} = A\vx'_i + \Bs u'_i,\quad i=0,1,\ldots,N-1,\\
   \vu &= \left[u'_0,u'_1,\ldots,u'_{N-1}\right]^\top
 \end{split}
\]
are plant state and input predictions. 
The matrices $P>0$, $Q>0$, and $W>0$ are chosen such that the feedback system is
asymptotically stable.
The procedure of choosing these matrices is presented in Section \ref{sec:stability}.

At each time instant $k\in\N_0$,
the controller uses the current state $\vx(k)$ to solve the above optimization with $\vx=\vx(k)$
thus providing the optimal control packet $\vu\bigl(\vx(k)\bigr)$.
This (possibly sparse) packet can be effectively compressed
before it is transmitted to the buffer
at the plant side. 

\section{Stability Analysis}
\label{sec:stability}

In this section, we show that if
\begin{itemize}
\item $\vv(k)=\vz$,
\item the matrices $P$, $Q$, and
$W$ in the proposed optimization (\ref{eq:opt0}) or (\ref{eq:optM})
are appropriately chosen,
\item and the maximum number of consecutive dropouts is
bounded,
\end{itemize}
then the NCS is asymptotically stable.
The proof is omitted due to limitation of space.

To consider the stability of the networked system
affected by packet dropouts, 
we follow akin to what was done in\cite{QueNes11}
and denote the
time instants where there are no
packet-dropouts, i.e., where
$d(k)=0$, as
\[
  \mathcal{K}=\{k_i\}_{i\in\N_0}\subseteq \N_0, \quad
  k_{i+1}>k_i,\; \forall i \in \N_0
\]
whereas the number of consecutive packet-dropouts is denoted via:
\begin{equation}
  m_i\eq k_{i+1}-k_i-1,\quad i\in\N_0.
  \label{eq:mi}
\end{equation}
Note that $m_i \geq 0$,
with equality if and only if no dropouts occur between instants $k_i$ and
$k_{i+1}$.

When packets are lost, the control system unavoidably operates in open-loop.
Thus, to ensure desirable properties of the networked control system, 
one would like the number of consecutive packet-dropouts to be bounded. 
In particular, to establish asymptotic stability,  
we  make the following assumption:
\footnote{If only stochastic properties are sought, 
then more relaxed assumptions can be used, see related work in \cite{QueOstNes10}.} 

\begin{ass}[Packet-dropouts]
\label{ass:dropouts}
The number of
consecutive packet-dropouts is  uniformly bounded by the prediction horizon
minus one,
that is, $m_i \leq N-1$,  $\forall i \in \N_0$.
We also assume that the first control packet $\vu\bigl(\vx(0)\bigr)$ is successfully transmitted,
that is, $m_0=0$. \hfs
\end{ass}

Theorem~\ref{thm:stability} stated below shows how to design
the matrices $P$, $Q$, and $W$ in~(\ref{eq:opt0})
to ensure asymptotic stability of the networked control system
in the presence of bounded packet dropouts.
Before proceeding, we introduce  the  matrices:
\[
 \begin{split}
  \Phi &\eq \begin{bmatrix}\Bs & 0 & \ldots & 0\\
  			A\Bs & \Bs & \ldots & 0\\
			\vdots & \vdots & \ddots & \vdots\\
			A^{N-1}\Bs & A^{N-2}\Bs & \ldots & \Bs
			\end{bmatrix},\quad
  = \begin{bmatrix}\Phi_0\\\Phi_1\\\vdots\\\Phi_{N-1}\end{bmatrix}\\
  \Phi_i &\eq \begin{bmatrix}A^iB & \dots & B & 0& \dots & 0\end{bmatrix},~i=0,1,\dots,N-1,\\
  \Upsilon &\eq \begin{bmatrix}A\\A^2\\\vdots\\A^N\end{bmatrix},~
  \bar{Q} \eq \mathrm{blockdiag}\{\underbrace{Q,\ldots,Q}_{N-1}, P\},
 \end{split}
\]
which allow us to re-write   (\ref{eq:opt0})
 in vector form via 
\begin{equation}
   \vu(\vx) = \argmin_{\vu\in\R^N} \|\vu\|_0 ~\text{subject to}~ \|G\vu-H\vx\|_2^2 \leq \vx^\top W \vx,
   \label{eq:optM}
\end{equation}
where
$G \eq \bar{Q}^{1/2}\Phi$ and  $H \eq -\bar{Q}^{1/2}\Upsilon$.

\begin{thm}[Asymptotic Stability]
\label{thm:stability}
Suppose that Assumption \ref{ass:dropouts} holds and that  the matrices $P$, $Q$, and $W$
are chosen by the following procedure:
\begin{enumerate}
\item Choose $Q>0$ arbitrarily.
\item Solve the following Riccati equation to obtain $P>0$:
\[
 P = A^\top P A - A^\top P\Bs(\Bs^\top P\Bs)^{-1}\Bs^\top PA + Q.
\]
\item Compute constants $\rho\in[0,1)$ and $c>0$ via
\[
 \begin{split}
   c_1 &\eq \max_{i=0,\ldots,N-1} \lambda_{\max} \bigl\{\Phi_i^\top P\Phi_i(G^\top G)^{-1}\bigr\}>0,\\
   \rho &\eq 1-\lambda_{\min}(QP^{-1}),\quad
   c \eq (1-\rho)^{-1}\bigl(1-\rho^N\bigr)c_1.
 \end{split}
\]
\item Choose $\E$ such that $0<\E<(1-\rho)P/c$.
\item Compute $W^\star=P-Q$ and set $W:=W^\star + \E$.
\end{enumerate}
Then the sparse control packets $\vu\bigl(\vx(k)\bigr)$, $k\in\N_0$,
which is the solution of the optimization (\ref{eq:opt0}) or (\ref{eq:optM})
with the above matrices,
lead to asymptotic stability of the networked control system.
\end{thm}
\section{Optimization via OMP}
\label{sec:relax}
In this section, we consider the optimization
\[
  (\Pzero):\qquad \min_{\vu\in\R^N} \|\vu\|_0~~ \text{subject to}~~ \|G\vu-H\vx\|_2^2 \leq \vx^\top W \vx.
\]
The optimization $(\Pzero)$
is in general extremely complex
since it requires a combinatorial search
that explores all possible sparse supports of $\vu\in\R^N$.
In fact, it is proved to be NP hard \cite{Nat95}.
For such problem, 
there have been proposed alternative algorithms
that are much more tractable than exhaustive search; see, e.g., the books 
\cite{Mal,Ela,StaMurFad}.

One approach to the combinatorial optimization is
an iterative greedy algorithm
called {\em Orthogonal Matching Pursuit} (OMP)
\cite{MalZha93,PatRezKri93}.
The algorithm is very simple and dramatically faster than
the exhaustive search.
In fact, assuming that $G\in\R^{m\times n}$ and
the solution $\vu^\ast$ of $(\Pzero)$ satisfies $\|\vu^\ast\|_0=k_0$,
then the OMP algorithm requires $O(k_0mn)$ operations,
while  exhaustive search requires $O(mn^{k_0}k_0^2)$
\cite{BruDonEla09}.\footnote{For control applications, OMP has recently been
  proposed for use 
in formation control in\cite{MasAnkVer11}.}
The OMP algorithm  for our control problem is shown in
Algorithm \ref{alg:OMP}.
In this algorithm,  
$\supp\{\vx\}$ is the support set of a vector 
$\vx=[x_1,x_2,\dots,x_n]^\top$,
that is, $\supp\{\vx\}=\{i:x_i\neq 0\}$,
and $\vg_j$ denotes the $j$-th column of the matrix $G$.
\begin{algorithm}[tb]
\caption{OMP for sparse control vector $\vu(\vx)$}
\label{alg:OMP}
\begin{algorithmic}
\REQUIRE $\vx\in\R^n$ \COMMENT{observed state vector}
\ENSURE $\vu(\vx)$ \COMMENT{sparse control packet}
\STATE $k:=0$.
\STATE $\vu[0]:=\vz$.
\STATE $\vr[0]:=H\vx-G\vu[0]=H\vx$.
\STATE $\Sr[0]:=\supp\{\vx[0]\}=\emptyset$.
\WHILE{$\|\vr[k]\|_2^2>\vx^\top W\vx$}
  \FOR{$j=1$ {\bf to} $N$}
    \STATE $\displaystyle z_j:=\frac{\vg_j^\top\vr[k]}{\|\vg_j\|_2^2}=\argmin_{z\in\R} \| \vg_jz - \vr[k] \|_2^2$.
    \STATE $e_j:=\|\vg_jz_j-\vr[k]\|_2^2$.
  \ENDFOR
  \STATE Find a minimizer $j_0\not\in\Sr[k]$ such that $e_{j_0}\leq e_{j}$, 
  	for all $j\not\in \Sr[k]$.
  \STATE $\Sr[k+1] := \Sr[k] \cup \{j_0\}$
  \STATE $\displaystyle \vu[k+1]:=\argmin_{\supp\{\vu\}=\Sr[k+1]}\|G\vu-H\vx\|_2^2$.
  \STATE $\vr[k+1]:=H\vx-G\vu[k+1]$.
  \STATE $k:=k+1$.
\ENDWHILE
\STATE {\bf return} $\vu(\vx)=\vu[k]$.
\end{algorithmic}
\end{algorithm}

Next, we study
stability of the NCS  with control packets
computed by  Algorithm \ref{alg:OMP}.

Since Algorithm \ref{alg:OMP}
always returns a feasible solution for $(\Pzero)$,
we have the following result based on Theorem \ref{thm:stability}.
\begin{thm}
\label{thm:OMP}
Suppose that Assumption \ref{ass:dropouts} holds and that the matrices $P$, $Q$,
and $W$ are chosen according to the procedure given
in Theorem \ref{thm:stability}.
Then, the control packets $\vu_{\OMP}(\vx(k))$, $k\in\N_0$ obtained by the OMP Algorithm \ref{alg:OMP}
provide an asymptotically stable NCS.
\end{thm}

Consequently, when compared to the method
used in \cite{nagque11a}, Algorithm~\ref{alg:OMP} has the following main advantages:
\begin{itemize}
\item it is simple and fast,
\item it returns  control packets that asymptotically stabilize the networked control system
\end{itemize}

We note that in conventional transform based compression
methods e.g., JPEG, the encoder maps the source signal 
into a domain where the
majority of the transform coefficients are approximately zero and only  few
coefficients carry significant information. One therefore only needs to encode
the few significant transform coefficients as well as their locations. 
In our case, on the other hand, we use the OMP algorithm to sparsify the control
signal in its original domain, which simplifies the decoder operations at the
plant side. To obtain a practical scheme for closed loop control, we employ
memoryless entropy-constrained scalar quantization of the 
non-zero coefficients of the sparse control signal and, in addition, send
information about the coefficient locations. We then show, through computer
simulations, that a significant bit-rate reduction is possible compared to when
performing memoryless entropy-constrained scalar quantization of the control signal
obtained by solving the standard quadratic control problem for PPC
as in \cite{Bem98}.%
\footnote{It is
  interesting to note that our proposed sparsifying controller could 
  also be useful in applications where there is a 
  \emph{setup cost} of the type found, for example, in inventory control; see,
  e.g.,\cite{berthi06}.  In such a case, it would
  be advantageous to have many zero control values.} 

\section{Simulation Studies}
\label{sec:examples}
To assess the effectiveness of
the proposed method, we consider the following continuous-time plant model:
\begin{equation}
 \begin{split}
  \dot{\vx}_c &= A_c \vx_c + B_c u,\\
  A_c&=\left[\begin{array}{rrrr}
   -1.2822&0&0.98&0\\ 0&0&1&0\\ -5.4293&0&-1.8366&0\\-128.2&128.2&0&0
  \end{array}\right],\\
  B_c&=\left[\begin{array}{r}
   -0.3\\
    0\\
   -17\\
    0\\
  \end{array}\right].
 \end{split}
 \label{eq:example-model}
\end{equation}
This model is a constant-speed approximation of some of the linealized
dynamics of a Cessna Citation 500 aircraft, when it is cruising at an altitude of 5000 (m)
and a speed of 128.2 (m/sec) \cite[Section 2.6]{Mac}.
To obtain a discrete-time model, we discretize \eqref{eq:example-model}
by the zero-order hold with sampling time $T_s = 0.5$ (sec)%
\footnote{
 This is done by MATLAB command {\tt c2d}.
}.
We set the  horizon length (or the packet size)  to $N=10$.
We choose the weighting matrix $Q$ in (\ref{eq:opt0}) as $Q=I$,
and choose the matrix $W$ according to the procedure shown in Theorem \ref{thm:stability}
with
$\E=\frac{2}{3}(1-\rho)P/c<(1-\rho)P/c$.

\subsection{Sparsity and Asymptotic Stability}
\label{sec:spars-asympt-stab}
We first simulate the NCS
in the noise-free case where $\vv(k)=\vz$. We consider the proposed method
using the OMP algorithm 
and also the $\elll$ optimization of\cite{nagque11a}:
\[
 (\Qone):\qquad \min_{\vu\in\R^N} \nu_1\|\vu\|_1 + \frac{1}{2}\|G\vu-H\vx\|_2^2,
\]
where $\nu_1$ is a positive constant.
To compare these two sparsity-promoting methods with traditional PPC approaches,
we also consider a finite-horizon quadratic cost function
\[
 (\Qtwo):\qquad \min_{\vu\in\R^N} \frac{\nu_2}{2}\|\vu\|_2^2 + \frac{1}{2}\|G\vu-H\vx\|_2^2,
\]
where $\nu_2$ is a positive constant, yielding
the $\ell^2$-optimal control
\[
 \vu_2(\vx) = (\nu_2I+G^\top G)^{-1}G^\top H\vx.
 \] 
To choose the regularization parameters $\nu_1$ in $(\Qone)$ and $\nu_2$ in $(\Qtwo)$,
we empirically compute the relation between each parameter and the control
performance, as measured by 
the $\ell^2$ norm of the state $\{\vx(k)\}_{k=0}^{99}$.
Fig.~\ref{fig:norm_vs_nu} shows this relation.
\begin{figure}[tbp]
\centering
\includegraphics[width=\linewidth]{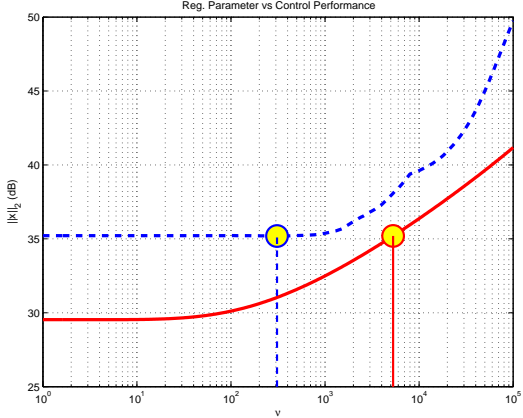}
\caption{Regularization parameters $\nu_i$ versus control performance $\|\vx\|_2$
for $(\Qone)$ (solid) and $(\Qtwo)$ (dash). The circles show the chosen parameters
$\nu_1$ and $\nu_2$.}
\label{fig:norm_vs_nu}
\end{figure}
By this figure, we first find the optimal parameter for $\nu_2>0$
that optimizes the control performance, i.e.,
$\nu_2=3.1\times 10^2$.
Then, we seek $\nu_1$
that gives the same control performance, namely,
$\nu_1=5.3\times 10^3$.
Furthermore,  we also investigate the ideal least-squares solution $\vu^\star(\vx)$
that minimizes $\|G\vu-H\vx\|_2$.

\par With  these parameters, we run 500 simulations
with randomly generated (Markovian) packet-dropouts that
satisfy Assumption \ref{ass:dropouts},
and with initial vector $\vx_0$
in which each element is independently sampled
from the normal distribution with mean 0
and variance 1. 
Fig.~\ref{fig:sparsity} shows the averaged sparsity of the obtained
control vectors.
\begin{figure}[tbp]
\centering
\includegraphics[width=\linewidth]{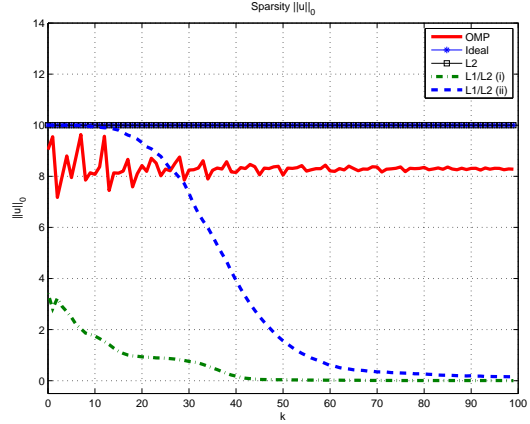}
\caption{
Sparsity of the control vectors by
ideal (star),
$\ell^2$ with regularization parameter $\nu_2=3.1\times 10^2$,
OMP (solid),
$\elll$ with
$\nu_1=5.3\times 10^3$
 (dash-dot)
 and
$\nu_1=5.3$ (dash).
}
\label{fig:sparsity}
\end{figure}
The $\elll$ optimization with 
 $\nu_1=5.3\times 10^3$
always produces
much sparser control vectors than those by OMP.
This property depends on how to choose the regularization parameter
$\nu_1>0$.
In fact, if we choose 
smaller $\nu_1=5.3$, the sparsity
changes as shown in Fig.~\ref{fig:sparsity}.
On the other hand, if we use a sufficiently large $\nu_1>0$,
then the control vector becomes $\vz$.
This is indeed the sparsest control, but leads
 to very poor control performance:
the state diverges until the control vector becomes
nonzero (see \cite{nagque11a}).

Fig.~\ref{fig:control_performance}
shows the averaged 2-norm of the state $\vx(k)$ as a function of $k$ for all 5 designs.
We see that, with exception of the $\elll$ optimization based PPC,  the NCSs are
nearly exponentially stable. In contrast, if the  $\elll$
optimization of\cite{nagque11a} is used, then only practical stability is observed.
The simulation  results are consistent with
Corollary \ref{thm:OMP} and our previous results in\cite{nagque11a}.
Note that the $\elll$ optimization with $\nu_1=5.3$
shows better performance than that with
$\nu_1=5.3\times 10^3$,
while $\nu_1=5.3\times 10^3$ gives a much sparser vector.
This shows a tradeoff between the performance and the sparsity.
\begin{figure}[tbp]
\centering
\includegraphics[width=\linewidth]{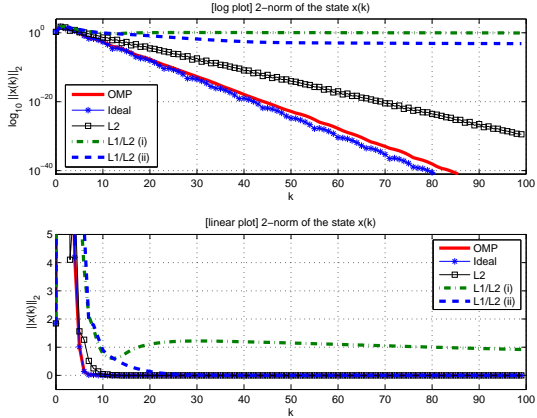}
\caption{2-norm of the state $\vx(k)$ for the four PPC designs:
log plot (above) and linear plot (below).}
\label{fig:control_performance}
\end{figure}

Fig.~\ref{fig:cpu_time} shows the associated computation times.
\begin{figure}[tbp]
\centering
\includegraphics[width=\linewidth]{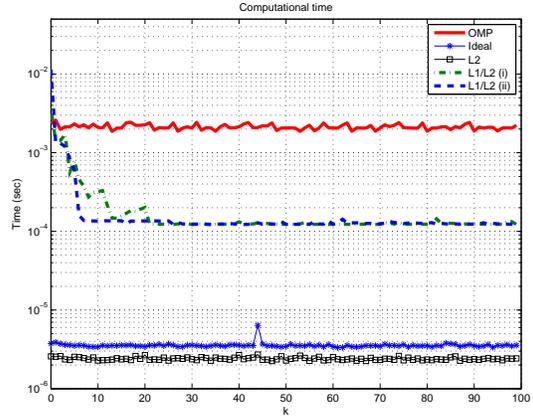}
\caption{Computational time.}
\label{fig:cpu_time}
\end{figure}
The $\elll$ optimization is faster than OMP in many cases.
Note that the ideal and the $\ell^2$ optimizations
are much faster, since
they require just one matrix-vector multiplication.


\subsection{Bit-rate Issues}
\label{sec:bitrate-issues}
We next investigate bit-rate aspects for a Gaussian plant noise process $\vv(k)$.
To keep the encoder and decoder simple, 
we will be using memoryless entropy-constrained uniform scalar quantization;
see \cite{CovTho}. 
Thus, the non-zero elements of the control vector are independently encoded 
using a scalar uniform quantizer followed by a scalar entropy coder. 
In the simulations, we choose the step size of the quantizer to be $\Delta = 0.001$, 
which results in negligible quantization distortion. 
We first run 1000 simulations with 100 time steps 
and use the obtained control vectors for designing entropy coders. 
A separate entropy coder is designed for each element in the control vector. 
For the first $N/2$ elements in the vector, 
we always use a quantizer followed by entropy coding. 
For the remaining $N/2$ elements, we only quantize and entropy code the non-zero elements. 
We then send additional $N/2$ bits indicating, which of the $N/2$ elements have
been encoded. 
The total bit-rate for each control vector is obtained as the sum of the codeword lengths 
for each individual non-zero codeword  $+N/2$ bits. 
For comparison, we use the same scalar quantizer with step size $\Delta=0.001$ 
and design entropy coders on the data obtained from the $\ell^2$ optimization. 
Since the control vectors in this case are non-sparse, we separately encode 
\emph{all} $N$ elements and sum the lengths of the individual codewords to obtain the total bit-rate. 
In both of the above cases, the entropy coders are Huffman coders. 
Moreover, the system parameters are initialized with different random seeds 
for the training and test situations, respectively.
The average rate per control vector for the OMP case is $55.36$ bits, 
whereas the average rate for the $\ell^2$ case is $112.57$ bits. Thus, due to sparsity, 
a $50.8$ percent bit-rate reduction is on average achieved.

\begin{figure}[tbp]
\centering
\includegraphics[width=\linewidth]{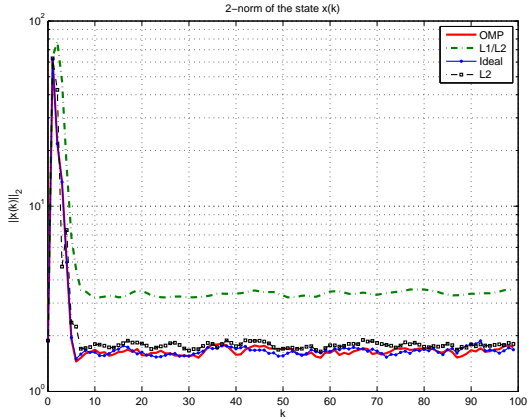}
\caption{2-norm of the state $\vx(k)$ with Gaussian noise.}
\label{fig:performance-2}
\end{figure}
\begin{figure}[tbp]
\centering
\includegraphics[width=\linewidth]{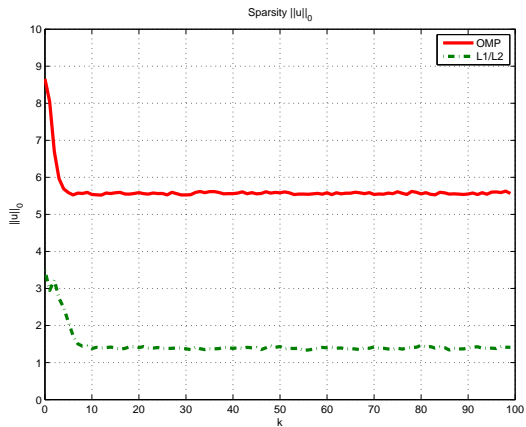}
\caption{Sparsity of the control vectors with Gaussian noise.}
\label{fig:sparsity-2}
\end{figure}

Fig.~\ref{fig:performance-2} shows the 2-norm of the state $\vx(k)$ and
Fig.~\ref{fig:sparsity-2} shows the sparsity.
We can also see the tradeoff between the performance and the sparsity
in this case.

\section{Conclusion}
\label{sec:conclusion}
We have studied a packetized
predictive control formulation  with a sparsity-promoting cost
function for error-prone rate-limited networked control system.
We have given sufficient conditions for asymptotic stability when the
controller is used over a network with bounded packet dropouts.
Simulation results indicate that the proposed controller provides
sparse control packets, thereby giving bit-rate reductions
when compared to the use of, more common, quadratic cost  functions.
Future work may include further study of performance aspects and the effect of
plant disturbances.  


\end{document}